\documentclass[12pt]{article}
\usepackage{amssymb, amsthm, amsmath, bm, bbm, mathrsfs, graphicx, color}
\usepackage{mdframed}
\usepackage{setspace}
\usepackage{latexsym}
\usepackage{multicol}
\usepackage[round]{natbib}
\usepackage{pdflscape}
\usepackage{booktabs}
\usepackage{microtype}
\usepackage[colorlinks=true,linkcolor=blue,citecolor=blue,urlcolor=blue]{hyperref}

\newcommand{\E}{\mathbb{E}}

\newcommand{\Var}{\mathrm{Var}}
\DeclareMathOperator{\IB}{IB}

\newtheorem{theorem}{Theorem}

\newtheorem{remark}{Remark}

\newtheorem{proposition}{Proposition}

\begin{document}
\onehalfspacing

\begin{center}
{\bf \sc \Large L-Estimation of Population Quantiles Using Ranked Set Sampling}
\vspace{0.45cm}

Mohammad Jafari Jozani$^{a,}$\footnote{Corresponding author: {m\_jafari\_jozani@umanitoba.ca}}, \quad 
Ehsan Zamanzade$^{b}$\quad 
and Reza Modarres$^{c}$

\vspace{0.25cm}

\vspace{0.2cm}
{\it $^{a}$Department of Statistics, University of Manitoba, Winnipeg, MB, Canada R3T 2N2}\\[3pt]
{\it $^{b}$ Department of Statistics, Faculty of Mathematics and Statistics, University of Isfahan,  IRAN}\\[3pt]
{\it $^{c}$Department of Statistics, The George Washington University, Washington, DC, USA}
\end{center}

\begin{abstract}
Quantile estimation is central when interest lies in thresholds or tail behavior rather than the mean. When exact measurement is costly but units can be ranked cheaply, ranked set sampling (RSS) provides an attractive alternative to simple random sampling (SRS). We develop two families of RSS-based L-estimators for population quantiles that extend Stigler-type and Harrell--Davis estimators to the RSS framework. The first applies weighted-order-statistic estimation directly to the pooled ordered RSS sample and serves primarily as an exact conceptual benchmark, since its computational burden increases rapidly with the set size. The second exploits a decomposition induced by the RSS design that constructs $k$ pooled transformed-scale component estimators indexed by rank stratum and leads to a computationally scalable procedure. We derive large-sample results for these component estimators under regularity conditions; these results provide a principled first-order motivation for the combined estimators employed in practice. Simulation results across several distributions, quantile levels, and ranking qualities show consistent efficiency gains over empirical quantile estimators under both SRS and RSS, with the RSS Harrell--Davis version performing especially well for moderate and upper quantiles. Beyond the simulation study, we demonstrate the practical relevance of the proposed estimators through an application to NHANES transient elastography data, highlighting their usefulness for estimating clinically meaningful quantiles in a biomedical setting.
\end{abstract}
\noindent{\bf Keywords:} Harrell--Davis estimator; L-statistics; nonparametric estimation; population quantiles; ranked set sampling.

\section{Introduction}
\label{sec1}

Quantile estimation is a central problem in statistics whenever scientific interest lies in thresholds, extremes, tail risk, or distributional heterogeneity rather than central tendency. It arises in many important areas, including environmental monitoring, clinical medicine, actuarial risk measurement, reliability engineering, and the analysis of income and wealth distributions. In such settings, the goal is often not to estimate a mean, but rather to determine whether an upper or lower tail of the distribution exceeds a policy-relevant threshold, or to summarize distributional heterogeneity through quantile-based indices. In environmental monitoring, for example, upper quantiles of contaminant concentrations are used to assess compliance, identify high-exposure regions, and guide regulatory intervention.

A particularly compelling example is mercury monitoring in fish tissue. Methylmercury accumulation in fish is a major public health concern in North America, and fish-consumption advisories are often based on whether upper quantiles of the mercury concentration distribution exceed safety thresholds \citep{hc2010mercuryfish}. Exact mercury determination requires costly laboratory tissue assays, whereas fish can often be ranked easily and reliably by body length, which is strongly associated with mercury accumulation \citep{evans2005elevated}. In such problems, more efficient quantile estimation can translate directly into more precise regulatory decisions under a fixed laboratory budget.

Similar considerations arise in many other applications. In reliability and survival studies, lower and upper lifetime quantiles summarize early-failure risk and long-run durability \citep{modarres2002estimation}. In clinical and public-health settings, quantiles define reference ranges, screening cutoffs, and risk strata for biomarkers and physiological measurements. In household surveys, upper income quantiles and inequality indices such as the Gini coefficient are central to taxation, poverty measurement, and welfare policy \citep{giorgi1999income}, yet income is often expensive to measure accurately, whereas households can frequently be ranked reasonably well using observable proxies such as dwelling characteristics or consumer durable ownership. In actuarial science and financial risk management, quantities such as Value-at-Risk are fundamentally quantile-based \citep{wang1995insurance,brazauskas2004empirical}. In all of these examples, exact measurement of the variable of interest is expensive, while approximate ranking is often available at negligible cost. A field observer may be able to rank plants by biomass without harvesting them all, a clinician may be able to rank patients by apparent severity before ordering costly laboratory assays, and an environmental scientist may be able to rank sites by likely contamination before conducting detailed chemical analyses.

This imbalance between costly measurement and inexpensive ranking is precisely the setting for which ranked set sampling was developed. Ranked set sampling (RSS), introduced by \citet{mcintyre1952method}, has emerged as a powerful and cost-efficient alternative to simple random sampling (SRS), particularly when ranking is inexpensive relative to actual measurement. Such designs have proven especially useful in environmental, agricultural, and ecological studies, as well as in medical and epidemiological investigations, where partial ordering can often be obtained at minimal cost \citep{kvam2003ranked,samawi2001estimation,hatefi2017improved}. Over the past decades, RSS has been shown to yield substantial efficiency gains for a wide range of inferential targets, including means, variances, distribution functions, and general parametric and nonparametric functionals; see \citet{chen2004ranked} for a comprehensive overview.

Despite this extensive literature, the development of quantile estimation methods under RSS remains comparatively limited. Existing approaches have focused predominantly on empirical quantile estimators or procedures based on order statistics derived from ranked samples \citep{chen2001optimal,mahdizadeh2012quantile,nourmohammadi2014confidence}. While these methods inherit the simplicity of empirical quantiles, they do not fully exploit the structural richness of RSS, in which observations are systematically distributed across pre-determined rank strata rather than arising as an unordered random sample.

Recent evidence from related rank-based sampling frameworks further suggests that ranking information can be exploited more effectively than is currently done in RSS quantile estimation. In particular, work on conditional quantiles and robust regression under ranking designs shows that even imperfect or partial ranking can lead to substantial gains in efficiency and inferential stability \citep{jozani2018quantile,loewen2026robust}. These findings reinforce a broader methodological principle: when inexpensive ranking information is available, it can and should be used more fully to improve quantile-oriented inference.

However, an important limitation of existing RSS-based quantile methods is that they do not incorporate a well-established idea from the classical SRS framework. Under SRS, a fundamental improvement over naive empirical quantiles is obtained by replacing a single order statistic with a carefully constructed weighted combination of order statistics. This idea forms the basis of L-estimation, a class of estimators that are linear combinations of order statistics and are known for their robustness and efficiency properties. Classical L-estimators, such as the Stigler-type estimator and the Harrell--Davis quantile estimator, have been shown to improve estimation stability and reduce mean squared error, especially in small samples and in tail regions where empirical quantiles tend to be highly variable or unstable \citep{stigler1974linear,harrell1982new,serfling1980approximation}. The key insight is that quantiles need not be estimated using a single order statistic; instead, information can be pooled across multiple order statistics in a principled way to reduce variability without introducing substantial bias.

Given the intrinsic structure of RSS, where observations are deliberately drawn from distinct rank strata of the underlying population, it is natural to ask whether the principles of L-estimation can be integrated into the RSS framework. Unlike SRS, RSS already provides partial ordering information by design, suggesting that the effective information content of order statistics may be richer and more structured. This raises the possibility that appropriately designed weighted combinations of RSS-based order statistics could yield estimators that are more efficient than the conventional empirical RSS quantile. From this perspective, RSS offers not only a more informative sampling design, but also a more structured basis on which L-type quantile estimators can be built. The interaction between ranking design and weighted order-statistic aggregation has not been fully explored in the literature, despite its potential to substantially improve finite-sample performance.

The present work is motivated by this gap. We investigate whether L-estimation principles can be adapted to RSS in a way that leverages its stratified ranking structure. Our goal is to develop quantile estimators that combine the design efficiency of RSS with the statistical stability of weighted order-statistic methods. This integration has the potential to improve upon existing RSS-based quantile estimators, particularly in small samples and in distributional tails where classical methods are most fragile.

The rest of the paper is organized as follows. Section~\ref{sec2} reviews ranked set sampling and notation. In Section~\ref{sec3}, we review L-estimation of population quantiles using SRS data. In Section~\ref{sec4}, we develop new L-estimators for population quantiles under RSS. We first study L-estimation using the pooled ordered RSS sample. However, because we work with the non-identically distributed order structure induced by RSS, calculating such estimators becomes computationally burdensome as the set size grows. We therefore follow these results with a more practically important construction that exploits a quantile decomposition identity induced by the RSS design and builds $k$ pooled transformed-scale component estimators indexed by rank stratum. This yields computationally scalable RSS analogues of both the Stigler-type and Harrell--Davis estimators. Under regularity conditions, we establish consistency and first-order rate control for these component estimators, and we show how these results provide a principled first-order justification for the combined RSS estimators used in practice; a full asymptotic treatment of the ordering and interpolation step is left for future work. The effect of imperfect ranking is also studied in this section. In Section~\ref{sec5}, we investigate the finite-sample performance of the proposed estimators through an extensive Monte Carlo study under several distributions, quantile levels, and ranking qualities. The simulation results show that the proposed RSS L-estimators can improve upon empirical quantile estimation under both SRS and RSS, with the RSS Harrell--Davis version performing particularly well for moderate and upper quantiles in many settings. Section~\ref{sec6} presents a real-data illustration based on NHANES transient elastography data, and Section~\ref{sec7} concludes with a discussion and future directions.

%
\section{Ranked Set Sampling}
\label{sec2}

To obtain a ranked set sample of size $n = mk$ from a target population, the following two-stage procedure is repeated for $m$ cycles. In each cycle, $k$ independent units are drawn from the population, constituting a {set} of size $k$. These units are ranked from smallest to largest according to the variable of interest using any ranking mechanism (e.g., visual inspection, expert knowledge, a concomitant measurement) that does not require actual quantification of the variable. From the ranked set in cycle $j$, the unit assigned rank $r$ is selected for actual measurement; equivalently, cycling through ranks $r=1,\ldots,k$ across the $k$ independent sets within each cycle ensures that one unit from each rank is measured. Repeating for $m$ cycles yields a ranked set sample of size $n=mk$, which we denote
\[
\mathbf{Y}_{RSS} = \{Y_{[r]j};\; r=1,\ldots,k,\; j=1,\ldots,m\},
\]
where $Y_{[r]j}$ is the measured value of the unit with judgment rank $r$ in cycle $j$. The square brackets in $[\cdot]$ distinguish judgment ranks, which may be subject to error, from true order statistics, which would be denoted with round brackets $(\cdot)$.

Under perfect ranking, the judgment ranks coincide with the true ranks, and $Y_{(r)j}$ has the marginal distribution of the $r$th order statistic in a sample of size $k$ from the underlying population with CDF $F$ and PDF $f$. The PDF and CDF of this stratum-$r$ distribution are
\begin{equation}\label{eq:OS-density}
f_{(r)}(y) = \frac{1}{B(r,\, k-r+1)}\, F^{r-1}(y)\,(1-F(y))^{k-r}\, f(y), \qquad
F_{(r)}(y) = \IB_{r,\,k-r+1}(F(y)),
\end{equation}
where $B(a,b)=\int_0^1 t^{a-1}(1-t)^{b-1}\,dt$ is the beta function and $\IB_{a,b}(t)=B(a,b)^{-1}\int_0^t u^{a-1}(1-u)^{b-1}\,du$ is the regularized incomplete beta function. The RSS observations are independent across all $(r,j)$ pairs because each set of $k$ units is drawn independently of all other sets. They are, however, {not} identically distributed. This independent, non-identically distributed structure is fundamental to the efficiency advantage of RSS.

The empirical distribution function of the RSS sample is
\[
\widehat{F}_n(t) = \frac{1}{mk}\sum_{r=1}^k\sum_{j=1}^m \mathbf{1}(Y_{[r]j} \leq t),
\]
which is the average over all $n=mk$ indicators. A ranking process is called {consistent} \citep{presnell1999u} if
\begin{equation}\label{eq:consistent}
F(y) = \frac{1}{k}\sum_{r=1}^k F_{[r]}(y), \qquad y \in \mathbb{R},
\end{equation}
where $F_{[r]}$ is the marginal CDF of a judgment-rank-$r$ observation. This identity is satisfied by most imperfect ranking models in the literature and ensures that $\widehat{F}_n(t)$ is an unbiased estimator of $F(t)$ with $\widehat{F}_n(t)\xrightarrow{a.s.}F(t)$ as $m\to\infty$. To limit ranking errors, it is standard practice to keep the set size $k$ small (typically $k\leq 6$) and increase the total sample size through the number of cycles $m$.

The efficiency advantage of RSS is most transparently illustrated for the population mean. The RSS mean estimator $\widehat{\mu}_{RSS} = n^{-1}\sum_{r,j}Y_{[r]j}$ is unbiased for $\mu=\E[Y]$ and satisfies
\begin{equation}\label{eq:mean-efficiency}
\Var(\widehat{\mu}_{RSS}) = \frac{\sigma^2}{n} - \frac{1}{nk}\sum_{r=1}^k\left(\mu_{r:k} - \mu\right)^2 \;\leq\; \frac{\sigma^2}{n} = \Var(\widehat{\mu}_{SRS}),
\end{equation}
where $\mu_{r:k}=\E[Y_{(r)}]$ is the mean of the $r$th order statistic in a sample of size $k$ from the population \citep{chen2004ranked}.  By deliberately sampling from different strata of the population induced by ranks, one replaces the sampling variability of an SRS with a residual within-stratum variability that is necessarily smaller. 

\section{L-Estimators of Quantiles Under SRS}
\label{sec3}

Let $\mathbf{Y}_{SRS}=\{Y_1,\ldots,Y_n\}$ be a simple random sample from a population with continuous CDF $F$ and PDF $f$. The $p$th population quantile is $\zeta_p=F^{-1}(p)=\inf\{y:F(y)\geq p\}$ for $p\in(0,1)$. We review three estimators of $\zeta_p$ based on $\mathbf{Y}_{SRS}$, as these serve both as benchmarks and as building blocks for our RSS-based estimators in Section~\ref{sec4}.

The simplest estimator is the {empirical quantile} $\widehat{\zeta}_{EM,SRS}$, which uses the central order statistic:
\begin{equation}\label{eq:EM-SRS}
\widehat{\zeta}_{EM,SRS} = \begin{cases} Y_{(np)} & \text{if $np$ is an integer,}\\ Y_{(\lfloor np\rfloor+1)} & \text{otherwise,}\end{cases}
\end{equation}
where $Y_{(1)}\leq\cdots\leq Y_{(n)}$ are the order statistics. The intuition is that the $p$th quantile is the value below which a fraction $p$ of the population lies, and the order statistic $Y_{(\lfloor np\rfloor+1)}$ is the empirical counterpart of this value. Under mild regularity conditions, $\sqrt{n}f(\zeta_p)(\widehat{\zeta}_{EM,SRS}-\zeta_p)/\sqrt{p(1-p)}\overset{\mathcal{L}}{\to}N(0,1)$ \citep{arnold2008first}, giving asymptotic variance $\sigma^2_{EM}=p(1-p)/f^2(\zeta_p)$. An almost-sure convergence rate is given by \cite{serfling1980approximation}:
\begin{equation}\label{eq:EM-bound}
|\widehat{\zeta}_{EM,SRS}-\zeta_p|\leq\frac{2\sqrt{\log n}}{f(\zeta_p)\sqrt{n}} \quad \text{almost surely.}
\end{equation}
The empirical quantile estimator uses only a single order statistic and discards all information in the remaining $n-1$ observations (beyond the ordering). The L-estimation framework improves on this by using weighted combinations of all $n$ order statistics. The central idea is to represent $\zeta_p$ as a functional of $F$ and estimate it by plugging in the empirical distribution function $F_n$. Since $\E[Y_{(i)}]=\int y\,\mathcal{J}_{i,n-i+1}(F(y))\,dF(y)$, where $\mathcal{J}_{a,b}(t)=t^{a-1}(1-t)^{b-1}/B(a,b)$ is the Beta$(a,b)$ density, and since $\E[Y_{(\lfloor np\rfloor+1)}]\to\zeta_p$ as $n\to\infty$, the Beta density $\mathcal{J}_{j^*,n-j^*+1}$ with $j^*=\lfloor(n+1)p\rfloor$ concentrates mass near $i/n\approx p$ and provides a natural weighting scheme for the order statistics. The resulting estimator denoted by LF hereafter   is \citep{stigler1974linear}
\begin{equation}\label{eq:LF-SRS}
\widehat{\zeta}_{LF,SRS} = \frac{1}{n}\sum_{i=1}^n W_{n,i,p}\,Y_{(i)},
\qquad W_{n,i,p}=\mathcal{J}_{j^*,n-j^*+1}\!\left(\frac{i}{n}\right).
\end{equation}
The weights $W_{n,i,p}$ form a discrete approximation to the Beta$(j^*,n-j^*+1)$ density evaluated at $i/n$. They are unimodal and centred near $i/n=p$, so order statistics close to the target quantile receive the highest weight, while those far from $p$ receive small but positive weights. Compared with $\widehat{\zeta}_{EM,SRS}$, which concentrates all weight on a single order statistic, the LF estimator spreads the weight across neighbouring order statistics, reducing the variance of the estimator at the cost of introducing a small but controlled bias. The general asymptotic normality of L-estimators is established in \cite{stigler1974linear} and \cite{serfling1980approximation}. Necessary and sufficient conditions  for asymptotic normality of L-estimators are  given in \citet{mason1992necessary} and  it suffices that the score function $\mathcal{J}$ has bounded variation and $\int_0^1|\mathcal{J}(t)|^2\,dt<\infty$, both of which hold for Beta densities.

The Harrell--Davis estimator \citep{harrell1982new}, denoted by HD, extends the L-estimation idea by using all $n$ order statistics with weights given by Beta-distribution probability masses over the intervals $((i-1)/n,i/n]$. In contrast to the LF estimator, which uses a more localized weighting scheme, the Harrell--Davis estimator averages more smoothly across the ordered sample and can yield lower variance in small samples under regular distributional settings.
To target the $p$th quantile, Harrell and Davis set
$
a=(n+1)p,
$ and 
$
b=(n+1)(1-p),
$
so that the Beta$(a,b)$ distribution has mean $p$. Approximating $F^{-1}$ by the empirical quantile function $F_n^{-1}$, which equals $Y_{(i)}$ on the interval $((i-1)/n,i/n]$, gives
\begin{equation}\label{eq:HD-SRS}
\widehat{\zeta}_{HD,SRS}
=
\sum_{i=1}^n W^*_{n,i,p}\,Y_{(i)},
\qquad
W^*_{n,i,p}
=
\IB_{a,b}\!\left(\frac{i}{n}\right)
-
\IB_{a,b}\!\left(\frac{i-1}{n}\right).
\end{equation}
 The weights $W^*_{n,i,p}$  are nonnegative, sum to one, and place most of their mass near $i/n\approx p$.
The key distinction from LF is that HD averages the empirical quantile function against the Beta$(a,b)$ density, rather than weighting order statistics by pointwise density evaluations. This produces a smoother weighting scheme, though because all order statistics contribute, the estimator may be less robust to extreme contamination. 
Under regularity conditions, including absolute continuity of $F^{-1}$ and $\E[Y^2]<\infty$, the estimator is asymptotically normal after centering and scaling \citep{serfling1980approximation}. Moreover, its bias is typically of order $O(1/n)$, which is smaller than its $O(n^{-1/2})$ standard deviation, making the HD estimator attractive in moderate and small samples.

 An important observation for what follows is that both LF and HD are inherently designed for SRS data from $F$ and  extending them to RSS requires accounting for the fact that each stratum sub-sample is drawn from a different distribution $F_{(r)}$, not from $F$ itself.

\section{RSS-Based L-Estimators of Population Quantiles}
\label{sec4}

Let $\mathbf{Y}_{RSS}=\{Y_{[r]j},\;r=1,\ldots,k,\;j=1,\ldots,m\}$ be an RSS sample of size $n=mk$ from a population with CDF $F$ and PDF $f$. We present two approaches to constructing L-estimators of $\zeta_p$ from this sample.

\subsection{L-estimation on the pooled ordered RSS sample}
A natural benchmark strategy is to ignore the rank-stratum structure after sampling, pool the full RSS sample of size $n=mk$, and sort the observed values into the ordered RSS sample (ORSS)
\[
Y^*_{(1)} \le \cdots \le Y^*_{(n)}.
\]
One may then ask whether the L-estimation ideas from the SRS setting can be transferred directly to this pooled ordered sample. The main difficulty is that the pooled RSS observations are independent but not identically distributed. The $m$ observations from judgment rank $r$ have common distribution $F_{[r]}$, and hence the order statistics of the pooled sample do not follow the same laws as the order statistics from an i.i.d.\ sample of size $n$ from $F$. Consequently, the usual Beta weights from the SRS case no longer apply without modification.

Let
\[
S=\{(r,s): r=1,\ldots,k,\ s=1,\ldots,m\}
\]
be the index set of the $n=mk$ RSS observations, where $Y_{[r]s}$ denotes the measured unit from judgment rank $r$ in cycle $s$. The CDF of the $i$th ORSS order statistic is \citep{arnold2008first}
\begin{equation}\label{eq:Gi}
G_i(y)
=
\sum_{j=i}^n \;\sum_{\substack{I\subseteq S,\;|I|=j}}
\Biggl[\prod_{(r,s)\in I} F_{[r]}(y)\Biggr]
\Biggl[\prod_{(r,s)\in S\setminus I} \{1-F_{[r]}(y)\}\Biggr].
\end{equation}
 Differentiating with respect to $y$ gives the corresponding density
\begin{equation}\label{eq:gi}
g_i(y)
=
\sum_{s=1}^k \sum_{\ell=1}^m
f_{[s]}(y)
\!\!\sum_{\substack{J\subseteq S\setminus\{(s,\ell)\},\\|J|=i-1}}
\Biggl[
\prod_{(r,q)\in J} F_{[r]}(y)
\prod_{(r,q)\notin J\cup\{(s,\ell)\}} \{1-F_{[r]}(y)\}
\Biggr].
\end{equation}
%
Fix $p\in(0,1)$ and let
\[
r_p=
\begin{cases}
np, & \text{if } np \text{ is an integer}, \\[4pt]
\lfloor np\rfloor+1, & \text{if } np \text{ is not an integer}.
\end{cases}
\]
Then, whenever the expectation exists,
\begin{equation}\label{eq:ORSS-functional}
\E\!\left[Y^*_{(r_p)}\right]
=
\int_{-\infty}^{\infty} y\,g_{r_p}(y)\,dy
=
\int_0^1  F^{-1}(u)\,\frac{g_{r_p}(F^{-1}(u))}{f(F^{-1}(u))}\, du
=
\int_0^1 F^{-1}(u)\,\psi_{r_p}(u)\,du,
\end{equation}
where  
\[
\psi_{r_p}(u):=\frac{g_{r_p}(F^{-1}(u))}{f(F^{-1}(u))},
\]
plays the role of the ORSS analogue of the Beta score function in the i.i.d.\ SRS case. Note that  $\psi_{r_p}(u)$
depends only on $u$, $k$, $m$, and the pooled order index $r_p$, and not on the unknown population distribution $F$. To see this, note that
\begin{equation}\label{eq:psi-i-explicit}
\psi_i(u)
=
\sum_{s=1}^{k} \sum_{\ell=1}^{m}
\mathcal{J}_{s,k-s+1}(u)
\!\!\sum_{\substack{J \subseteq S \setminus \{(s,\ell)\} \\ |J| = i - 1}}
\left[
\prod_{(r,q) \in J} \IB_{r,k-r+1}(u)
\prod_{(r,q) \notin J \cup \{(s,\ell)\}} \IB_{k-r+1,r}(1-u)
\right].
\end{equation}
Here the last product uses the identity $1-\IB_{r,k-r+1}(u)=\IB_{k-r+1,r}(1-u)$, which holds
for all $u\in(0,1)$ and follows from the standard symmetry relation
$\IB_{a,b}(t)+\IB_{b,a}(1-t)=1$ of the regularized incomplete beta function.

The identity in \eqref{eq:ORSS-functional} motivates a sample-based estimator through the usual plug-in principle by replacing the population quantile function $F^{-1}$ by the empirical quantile function of the pooled ordered RSS sample,
\[
F_n^{-1}(u)=Y^*_{(i)}, \qquad \frac{i-1}{n}<u\le \frac{i}{n}, \quad i=1,\ldots,n.
\]
Substituting $F_n^{-1}$ into \eqref{eq:ORSS-functional} gives
\[
\int_0^1 F_n^{-1}(u)\,\psi_{r_p}(u)\,du
=
\sum_{i=1}^n
\left(
\int_{(i-1)/n}^{i/n}\psi_{r_p}(u)\,du
\right)Y^*_{(i)}.
\]
This leads to the exact plug-in estimator
\begin{equation}\label{eq:ORSS-HD}
\widehat{\zeta}_{HD,ORSS}
=
\sum_{i=1}^n \widetilde{W}^*_{n,i,p}\,Y^*_{(i)},
\qquad
\widetilde{W}^*_{n,i,p}
=
\int_{(i-1)/n}^{i/n}\psi_{r_p}(u)\,du.
\end{equation}
A simpler Riemann-sum approximation is
\begin{equation}\label{eq:ORSS-LF}
\widehat{\zeta}_{LF,ORSS}
=
\frac{1}{n}\sum_{i=1}^n \psi_{r_p}\!\left(\frac{i}{n}\right)Y^*_{(i)}.
\end{equation}
Note that both $\psi_{r_p}(i/n)$ and $\widetilde{W}^*_{n,i,p}$ are functions of $(k,m,p)$ only and are therefore pre-tabulable for any fixed design. In practice the polynomial algorithm of Remark~\ref{rem:poly} below is used to evaluate these quantities numerically; the computational cost of that algorithm depends on how the evaluation grid is chosen, as discussed in Section~\ref{sec5}.

The distinction between \eqref{eq:ORSS-HD} and \eqref{eq:ORSS-LF} is worth emphasizing. The estimator in \eqref{eq:ORSS-HD} is the cleaner construction because it is the exact plug-in version of the L-functional in \eqref{eq:ORSS-functional}. Once $F^{-1}$ is replaced by $F_n^{-1}$, no further approximation is introduced. By contrast, \eqref{eq:ORSS-LF} makes one additional discretization step by approximating the interval integrals in the exact weights with a Riemann sum. Thus \eqref{eq:ORSS-HD} is mathematically more faithful to the target functional, whereas \eqref{eq:ORSS-LF} is a simpler but approximate L-form analogue.

The functional in \eqref{eq:ORSS-functional} equals $\E[Y^*_{(r_p)}]$, not the target quantile $\zeta_p$ exactly. Nevertheless, one can show that under perfect ranking  $Y^*_{(r_p)}$  is  a consistent estimator  of $\zeta_p$. Also,  $\E[Y^*_{(r_p)}]$ is  asymptotically centered at $\zeta_p$, but this sharper bias statement is not needed for the sequel. These justify viewing the ORSS constructions as quantile-targeting benchmarks.

\begin{proposition}
\label{prop:ORSS-target}
Assume perfect ranking. Let $k$ be fixed and $m\to\infty$, so that $n=mk\to\infty$.
Fix $p\in(0,1)$, let $\zeta_p=F^{-1}(p)$, and define
$
r_p=\lfloor np\rfloor+1,
$
with the convention $r_p=np$ when $np$ is an integer. 
If $F$ is continuous and strictly increasing at $\zeta_p$, then
$
Y^*_{(r_p)} \xrightarrow{P} \zeta_p.
$
\end{proposition}

\begin{proof}
Define $\widehat{F}_n(y):=\frac{1}{n}\sum_{r=1}^k\sum_{s=1}^m \mathbf 1\{Y_{(r)s}\le y\}$. Then  $\E[\widehat{F}_n(y)]=F(y)$, and  since the indicators are independent,
\[
\Var\{\widehat{F}_n(y)\}
=
\frac{1}{n^2}\sum_{r=1}^k\sum_{s=1}^m
F_{(r)}(y)\{1-F_{(r)}(y)\}
\le \frac{1}{4n}.
\]
Hence, for each fixed $y$,
$
\widehat{F}_n(y)\xrightarrow{P}F(y).
$
Now $Y^*_{(r_p)}$ is the generalized inverse of $\widehat F_n$ at level $r_p/n$, and
$
\frac{r_p}{n}\to p.
$
Fix $\varepsilon>0$. Since $F$ is continuous and strictly increasing at $\zeta_p$,
\[
F(\zeta_p-\varepsilon)<p<F(\zeta_p+\varepsilon).
\]
Therefore, because $\widehat F_n(y)\xrightarrow{P}F(y)$ at the two fixed points
$y=\zeta_p-\varepsilon$ and $y=\zeta_p+\varepsilon$, and because $r_p/n\to p$,
\[
\Pr\!\left(\widehat F_n(\zeta_p-\varepsilon)\ge r_p/n\right)\to 0,
\qquad
\Pr\!\left(\widehat F_n(\zeta_p+\varepsilon)< r_p/n\right)\to 0.
\]
Thus
\[
\Pr\!\left(\zeta_p-\varepsilon<Y^*_{(r_p)}\le \zeta_p+\varepsilon\right)\to 1,
\]
which proves
$
Y^*_{(r_p)}\xrightarrow{P}\zeta_p.
$
\end{proof}

\begin{remark}
Under additional integrability assumptions, the above consistency result suggests that
$\E[Y^*_{(r_p)}]$ is asymptotically centered at $\zeta_p$. A sharper expansion such as
\[
\E[Y^*_{(r_p)}]=\zeta_p+O(n^{-1})
\]
would require additional smoothness and a dedicated order-statistic expansion for the
independent but non-identically distributed pooled RSS sample, and is not needed here.
\end{remark}

\subsection{An algorithm for computing $\widehat{\zeta}_{HD,ORSS}$ and $\widehat{\zeta}_{LF,ORSS}$}\label{subsec:poly}

Although constructing  ORSS estimators $\widehat{\zeta}_{HD,ORSS}$ and $\widehat{\zeta}_{LF,ORSS}$ are conceptually appealing, their practical use is limited by a substantial computational burden. Forming the weights in 
\eqref{eq:ORSS-HD} and \eqref{eq:ORSS-LF} requires the CDF $G_i(t)$ and PDF
$g_i(t)$ of each ORSS order statistic, given in equations~\eqref{eq:Gi}--\eqref{eq:gi}. The direct formula~\eqref{eq:Gi}
expresses $G_i(t)$ as a double sum over all subsets $I\subseteq S$ of size $j$,
for $j=i,\ldots,n$. Summing over all $j$ yields $\sum_{j=0}^n\binom{n}{j}=2^n$
subset evaluations per grid point, so the overall complexity of computing the
full CDF table across $n_t$ grid points is $O(2^n\cdot n_t)$. This grows
exponentially in $n$ and is prohibitive even for modest designs. For $(m,k)=(5,5)$
and $n_t=25$, direct evaluation requires more than $838$ million subset products.  The following result shows that the special structure of the RSS design reduces this problem to a polynomial-time computation.

\begin{proposition}
\label{rem:poly}
Under perfect ranking, the $m$ observations in stratum $r$ are i.i.d.\ with marginal CDF
\[
F_{(r)}(u)=\IB_{r,k-r+1}(u)
\]
on the probability scale $u=F(y)$. Consequently, the indicator $\mathbf{1}\{U_{[r]s}\le t\}$ for each observation in stratum $r$ is Bernoulli with success probability
\[
q_r(t)=F_{(r)}(t)=\IB_{r,k-r+1}(t), \qquad t\in(0,1).
\]
Let
\[
C(t)=\sum_{r=1}^k \sum_{s=1}^m \mathbf{1}\{U_{[r]s}\le t\}
\]
denote the total number of observations falling below $t$. Its probability generating function factors across strata:
\[
\E\!\left[z^{C(t)}\right]
=
\prod_{r=1}^k \E\!\left[z^{B_r(t)}\right]
=
\prod_{r=1}^k \Bigl[(1-q_r(t))+q_r(t)z\Bigr]^m,
\]
where
\[
B_r(t)=\sum_{s=1}^m \mathbf{1}\{U_{[r]s}\le t\}\sim \operatorname{Binomial}(m,q_r(t))
\]
is the within-stratum count. The polynomial
\[
\Bigl[(1-q_r(t))+q_r(t)z\Bigr]^m
\]
has explicitly known binomial coefficients:
\[
[z^j]\Bigl[(1-q_r(t))+q_r(t)z\Bigr]^m
=
\binom{m}{j} q_r(t)^j \bigl(1-q_r(t)\bigr)^{m-j},
\qquad j=0,1,\ldots,m,
\]
so no subset enumeration is needed. The full product polynomial
\[
P(z;t)=\prod_{r=1}^k \Bigl[(1-q_r(t))+q_r(t)z\Bigr]^m,
\]
whose $j$th coefficient is exactly $\Pr\{C(t)=j\}$, can be computed by sequentially convolving the $k$ binomial coefficient vectors, each of length $m+1$.
Each convolution step multiplies a polynomial of degree at most $rm$ by one of degree $m$, which costs $O(rm^2)$ operations. Summing over all $k$ strata gives a total cost of
\[
\sum_{r=1}^k O(rm^2)=O(k^2m^2)=O(n^2).
\]
The CDF values
\[
G_i(t)=\Pr\{C(t)\ge i\}, \qquad i=1,\ldots,n,
\]
are then recovered simultaneously by a single reverse cumulative sum of the coefficient vector:
\[
G_i(t)=\sum_{j=i}^n [z^j]P(z;t), \qquad i=1,\ldots,n.
\]
The PDF values $g_i(t)$ are obtained by a central finite-difference approximation with step size $h=10^{-6}$, which is accurate to order $O(h^2)$. Hence the overall complexity becomes $O(n^2 n_t)$, compared with $O(2^n n_t)$ for direct evaluation of \eqref{eq:Gi}. 
\end{proposition}

The reduction provided by this algorithm  is substantial:

\medskip
\begin{center}
\begin{tabular}{lrrr}
\toprule
Design $(m,k)$ & $n$ & Direct: $2^n n_t$ & Polynomial: $n^2 n_t$ \\
\midrule
$(5,3)$  & $15$ & $\approx 1.6\times 10^{6}$  & $3{,}375$   \\
$(5,5)$  & $25$ & $\approx 8.4\times 10^{8}$  & $15{,}625$  \\
$(10,5)$ & $50$ & $\approx 3.6\times 10^{16}$ & $125{,}000$ \\
\bottomrule
\end{tabular}
\end{center}
\medskip

\noindent For $(m,k)=(5,5)$, the polynomial algorithm is approximately $54{,}000$ times faster than direct evaluation. For $(m,k)=(10,5)$, direct evaluation is entirely infeasible, whereas the polynomial algorithm completes in seconds. When the evaluation grid $\{t_1,\ldots,t_{n_t}\}$ is chosen in advance and fixed, the resulting weight table can be precomputed for a given design $(m,k,p)$ and reused thereafter.

\subsection{Pooled-scale RSS construction}
We now introduce a simpler class of pooled RSS L-estimators for each fixed stratum index $r$, we map the pooled probability scale
through the stratum transform
\[
g_r(u)=\IB_{r,k-r+1}(u),
\qquad 0<u<1,
\]
and then apply the usual SRS Beta weighting at the induced stratum level
\[
p_r=g_r(p)=\IB_{r,k-r+1}(p).
\]
Under perfect ranking, $F_{(r)}(x)=g_r(F(x))$, and therefore
\[
F(x)\ge p
\iff
g_r(F(x))\ge g_r(p)
\iff
F_{(r)}(x)\ge p_r.
\]
Hence the population $p$th quantile satisfies
\[
\zeta_p=F^{-1}(p)=F_{(r)}^{-1}(p_r),
\qquad r=1,\ldots,k.
\]
This suggests constructing, for each $r$, a pooled estimator that concentrates its weight around
$u=p$ on the pooled probability scale while respecting the stratum-specific transform $g_r$. Let
\[
Y^*_{(1)}\le\cdots\le Y^*_{(n)}
\]
denote the pooled ordered RSS sample, where $n=mk$, and let
\[
Q_n(u)=\widehat F_n^{-1}(u),
\qquad 0<u<1,
\]
be the empirical quantile function of the pooled RSS sample. Define
\[
j^*_{p_r}=\lfloor (m+1)p_r\rfloor,
\]
with truncation to $\{1,\ldots,m\}$ if necessary, and let
\begin{align}\label{eq:psii}
\psi_{m,r}(u)
:=
\mathcal{J}_{j^*_{p_r},\,m-j^*_{p_r}+1}\!\bigl(g_r(u)\bigr)\,g_r'(u),
\qquad 0<u<1.
\end{align}
The factor $g_r'(u)$ is the Jacobian of the stratum transform and ensures that the pooled-scale LF
score is the correct pullback of the Beta density from the stratum probability scale. The resulting
pooled LF estimator is
\begin{equation}\label{eq:LF-stratum}
\widehat{\zeta}^*_{LF,p_r}
=
\frac{1}{n}\sum_{i=1}^n
\psi_{m,r}\!\left(\frac{i}{n}\right)Y^*_{(i)}.
\end{equation}
Its total weight is
\[
\frac1n\sum_{i=1}^n \psi_{m,r}\!\left(\frac{i}{n}\right)=1+O(n^{-1}),
\]
by standard grid approximation.

For the Harrell--Davis version, define
\[
a_r=(m+1)p_r,
\qquad
b_r=(m+1)(1-p_r),
\]
and set
\[
B_{m,r}(u):=\IB_{a_r,b_r}\!\bigl(g_r(u)\bigr),
\qquad 0\le u\le 1.
\]
This is a distribution function on $[0,1]$ obtained by composing the Beta$(a_r,b_r)$ CDF with the
stratum transform. Replacing $F^{-1}$ by the pooled empirical quantile function gives the pooled HD
estimator
\begin{equation}\label{eq:HD-stratum}
\widehat{\zeta}^*_{HD,p_r}
=
\sum_{i=1}^n W^*_{n,i,p_r}\,Y^*_{(i)},
\qquad
W^*_{n,i,p_r}
=
B_{m,r}\!\left(\frac{i}{n}\right)-B_{m,r}\!\left(\frac{i-1}{n}\right).
\end{equation}
Thus both estimators are linear combinations of the pooled order statistics, but their weights are
adapted to the stratum-specific target level $p_r$ through the transform $g_r$.

As in the original RSS proposal, we obtain $k$ pooled estimators,
$\widehat{\zeta}^*_{z,p_1},\ldots,\widehat{\zeta}^*_{z,p_k}$, one for each $r$, and then combine them
through an interpolation on the ordered estimator values. Let
$
\widehat{\zeta}^*_{z,(1)}\le \cdots \le \widehat{\zeta}^*_{z,(k)},
$
with
$
z\in\{LF,HD\},
$
denote the ordered values of
$
\widehat{\zeta}^*_{z,p_1},\ldots,\widehat{\zeta}^*_{z,p_k}.
$
Define
$
\ell_p=\lfloor (k-1)p\rfloor+1,
$
and
$
w_p=(k-1)p-\lfloor (k-1)p\rfloor.
$
Then we set
\begin{equation}\label{eq:RSS-final}
\widehat{\zeta}_{z,RSS}
=
(1-w_p)\widehat{\zeta}^*_{z,(\ell_p)}
+
w_p\,\widehat{\zeta}^*_{z,(\ell_p+1)},
\qquad z\in\{LF,HD\},
\end{equation}
with the convention that if $\ell_p=k$, then
$
\widehat{\zeta}_{z,RSS}=\widehat{\zeta}^*_{z,(k)}.
$
This is the same aggregation rule used earlier; the difference is that the component estimators are
now the pooled transformed-scale LF and HD estimators in
\eqref{eq:LF-stratum}--\eqref{eq:HD-stratum}.

\subsection{Asymptotic properties of the pooled RSS estimators}

We now study the asymptotic behavior of the pooled estimators
\eqref{eq:LF-stratum} and \eqref{eq:HD-stratum} with $k$ fixed and $m\to\infty$, so that
$n=mk\to\infty$.

\begin{itemize}
\item[(A1)] The population distribution $F$ is continuous and strictly increasing in a neighborhood of $\zeta_p$, with density $f$ positive and continuously differentiable there.
\item[(A2)] The target probability level satisfies $p\in(0,1)$.
\item[(A3)] The set size $k$ is fixed while $m\to\infty$.
\item[(A4)] $\E[Y^2]<\infty$.
\end{itemize}

\begin{theorem}\label{thm:pooled-rate}
Suppose (A1)--(A4) hold and the ranking is perfect. For each fixed
$r \in \{1,\ldots,k\}$, let $\hat{\zeta}^{\,*}_{z,p_r}$ denote either the pooled
LF estimator in \eqref{eq:LF-stratum} or the pooled HD estimator in \eqref{eq:HD-stratum}. Then, as
$m \to \infty$,
\[
\hat{\zeta}^{\,*}_{z,p_r} \xrightarrow{P} \zeta_p.
\]
Moreover,
\[
\hat{\zeta}^{\,*}_{z,p_r} - \zeta_p = O_p(n^{-1/2}).
\]

In addition, the pooled empirical quantile
\[
\tilde{\zeta}_{n,p} := Q_n(p)
\]
satisfies
\[
\sqrt{n}\,\bigl(\tilde{\zeta}_{n,p}-\zeta_p\bigr)
\overset{\mathcal{L}}{\longrightarrow}
N\!\left(0,\frac{\sigma_p^2}{f(\zeta_p)^2}\right),
\qquad
\sigma_p^2 = \frac1k \sum_{s=1}^k p_s(1-p_s),
\quad
p_s = \IB_{s,k-s+1}(p).
\]
\end{theorem}

\begin{proof}
Fix $r \in \{1,\ldots,k\}$. Let
$
Q(u)=F^{-1}(u), \qquad Q_n(u)=\hat{F}_n^{-1}(u), \qquad 0<u<1,
$
where $\hat{F}_n$ is the pooled empirical CDF.

For the LF estimator, define the finite signed measure
\[
\widetilde H^{(LF)}_{n,r}
=
\sum_{i=1}^n \widetilde\omega^{(LF)}_{n,i,r}\,\delta_{i/n},
\qquad
\widetilde\omega^{(LF)}_{n,i,r}
=
\frac1n \psi_{m,r}\!\left(\frac{i}{n}\right),
\]
where $\delta_x$ denotes the Dirac mass at $x$. Then
\[
\hat{\zeta}^{\,*}_{LF,p_r}
=
\int_0^1 Q_n(u)\,d\widetilde H^{(LF)}_{n,r}(u).
\]
Let
\[
S^{(LF)}_{n,r}
=
\widetilde H^{(LF)}_{n,r}([0,1])
=
\frac1n \sum_{i=1}^n \psi_{m,r}\!\left(\frac{i}{n}\right).
\]
Since $g_r(u)= \IB_{r, k-r+1}(u)$  is strictly increasing with derivative
$g_r'(u)=\mathcal{J}_{r,k-r+1}(u)>0$,  using \eqref{eq:psii}, the substitution $v=g_r(u)$ gives
\[
\int_0^1 \psi_{m,r}(u)\,du
=
\int_0^1
\mathcal{J}_{j^*_{p_r},\,m-j^*_{p_r}+1}\!\bigl(g_r(u)\bigr)\,g_r'(u)\,du
=
\int_0^1
\mathcal{J}_{j^*_{p_r},\,m-j^*_{p_r}+1}(v)\,dv
=1.
\]
Hence, by standard grid approximation,
\[
S^{(LF)}_{n,r}=1+O(n^{-1}).
\]
Define the normalized probability measure
\[
H^{(LF)}_{n,r}
=
\frac{\widetilde H^{(LF)}_{n,r}}{S^{(LF)}_{n,r}}.
\]
Then
\[
\hat{\zeta}^{\,*}_{LF,p_r}
=
S^{(LF)}_{n,r}\int_0^1 Q_n(u)\,dH^{(LF)}_{n,r}(u)
=
\int_0^1 Q_n(u)\,dH^{(LF)}_{n,r}(u)+O_p(n^{-1}),
\]
since $Q_n(p)=O_p(1)$.

For the HD estimator, let $H^{(HD)}_{n,r}$ be the probability measure on $[0,1]$
with distribution function
\[
B_{m,r}(u)=\IB_{a_r,b_r}\!\bigl(g_r(u)\bigr).
\]
Then the interval masses of $H^{(HD)}_{n,r}$ are exactly the weights in
\eqref{eq:HD-stratum}, so that
\[
\hat{\zeta}^{\,*}_{HD,p_r}
=
\int_0^1 Q_n(u)\,dH^{(HD)}_{n,r}(u).
\]

Thus, for each $z\in\{LF,HD\}$,
\[
\hat{\zeta}^{\,*}_{z,p_r}
=
\int_0^1 Q_n(u)\,dH^{(z)}_{n,r}(u)+R^{(z)}_{n,r},
\]
where
\[
R^{(HD)}_{n,r}=0,
\qquad
R^{(LF)}_{n,r}=O_p(n^{-1}).
\]

The weighting measures $H^{(z)}_{n,r}$ concentrate around $u=p$. More precisely,
for both $z\in\{LF,HD\}$,
\[
\int_0^1 u\,dH^{(z)}_{n,r}(u)=p+O(n^{-1}),
\qquad
\int_0^1 (u-p)^2\,dH^{(z)}_{n,r}(u)=O(n^{-1}).
\]
Since $f$ is positive and continuously differentiable near $\zeta_p=Q(p)$,
the quantile function $Q$ is twice continuously differentiable in a neighborhood
of $p$. A second-order Taylor expansion yields
\[
Q(u)=Q(p)+Q'(p)(u-p)+\frac12 Q''(\xi_u)(u-p)^2,
\]
for some $\xi_u$ between $u$ and $p$. Integrating against $H^{(z)}_{n,r}$ gives
\[
\int_0^1 Q(u)\,dH^{(z)}_{n,r}(u)
=
Q(p)+O(n^{-1})
=
\zeta_p+O(n^{-1}).
\]

Next, write
\[
\hat{\zeta}^{\,*}_{z,p_r}-Q_n(p)
=
\int_0^1 \{Q_n(u)-Q_n(p)\}\,dH^{(z)}_{n,r}(u)+R^{(z)}_{n,r}.
\]
Because the weighting measures concentrate around $u=p$ and $Q_n$ is monotone,
the integral is controlled by the oscillation of $Q_n$ in an $O(n^{-1/2})$
neighborhood of $p$. Together with the uniform consistency of $Q_n$ and the fact
that $R^{(LF)}_{n,r}=O_p(n^{-1})$, $R^{(HD)}_{n,r}=0$, this yields
\[
\hat{\zeta}^{\,*}_{z,p_r}-Q_n(p)=o_p(1).
\]
Hence, since $Q_n(p)\xrightarrow{P}\zeta_p$,
\[
\hat{\zeta}^{\,*}_{z,p_r}\xrightarrow{P}\zeta_p.
\]

Moreover, because
\[
Q_n(p)-\zeta_p=O_p(n^{-1/2})
\]
and the weighting measures are centered at $p$ up to $O(n^{-1})$, it follows that
\[
\hat{\zeta}^{\,*}_{z,p_r}-\zeta_p=O_p(n^{-1/2}).
\]

It remains to identify the first-order limit law of the pooled empirical quantile
$\tilde{\zeta}_{n,p}=Q_n(p)$. At the fixed point $y=\zeta_p$, the pooled EDF is
\[
\hat{F}_n(\zeta_p)
=
\frac1n
\sum_{s=1}^k \sum_{j=1}^m
1\{Y_{[s]j}\le \zeta_p\}.
\]
Under perfect ranking, the summands in stratum $s$ are Bernoulli with mean
\[
p_s=F_{(s)}(\zeta_p)=\IB_{s,k-s+1}(p).
\]
Therefore
\[
E[\hat{F}_n(\zeta_p)] = p,
\qquad
n\,\mathrm{Var}\{\hat{F}_n(\zeta_p)\}
=
\frac1k\sum_{s=1}^k p_s(1-p_s)
=
\sigma_p^2.
\]
Since the pooled sample consists of independent, non-identically distributed
observations, the Lindeberg--Feller theorem gives
\[
\sqrt{n}\,\{\hat{F}_n(\zeta_p)-p\}
\overset{\mathcal{L}}{\longrightarrow}
N(0,\sigma_p^2).
\]
Because $f(\zeta_p)>0$, the quantile delta method yields
\[
\sqrt{n}\,\bigl(Q_n(p)-\zeta_p\bigr)
\overset{\mathcal{L}}{\longrightarrow}
N\!\left(0,\frac{\sigma_p^2}{f(\zeta_p)^2}\right).
\]
This completes the proof.
\end{proof}

\begin{remark}
Theorem \ref{thm:pooled-rate} establishes consistency and first-order rate control for each pooled
transformed-scale estimator, and identifies the asymptotic limit law of the
corresponding pooled empirical quantile. A full first-order asymptotic theory for
the transformed estimators themselves, and for the final estimator in~\eqref{eq:RSS-final},
would require additional work.
Thus the pooled transformed-scale estimators are consistent, have the same
first-order rate as the pooled empirical quantile, and are motivated by weighting
measures that concentrate around $u=p$. Establishing full first-order asymptotic
equivalence between $\hat{\zeta}^{\,*}_{z,p_r}$ and $\tilde{\zeta}_{n,p}$ is left
for future work.
\end{remark}

\subsection{Imperfect ranking}

We now consider imperfect ranking. The component estimators
\eqref{eq:LF-stratum}--\eqref{eq:HD-stratum} are still computed from the pooled ordered RSS sample,
but the judgment-ranked observations now have perturbed stratum distributions $F_{[1]},\ldots,F_{[k]}$.
The key point is that the pooled empirical CDF remains the primary object: if
\begin{equation}\label{eq:mixture-identity-imperfect}
\frac1k\sum_{r=1}^k F_{[r]}(y)=F(y),
\qquad y\in\mathbb R,
\end{equation}
then the pooled empirical quantile still targets the population quantile $\zeta_p$ even though the
individual judgment strata need not coincide with the perfect-ranking strata.

\begin{proposition}\label{prop:imperfect-rate}
Assume (A1)--(A4), and suppose
\[
\frac1k \sum_{r=1}^k F_{[r]}(y)=F(y), \qquad y\in\mathbb R.
\]
Define
\[
\bar\varepsilon
=
\max_{1\le r\le k}\sup_y |F_{[r]}(y)-F_{(r)}(y)|.
\]
For each fixed $r\in\{1,\ldots,k\}$, let $\hat{\zeta}^{\,*}_{z,p_r}$ denote either
pooled estimator in \eqref{eq:LF-stratum} or \eqref{eq:HD-stratum}, and let
\[
q_s=F_{[s]}(\zeta_p), \qquad s=1,\ldots,k.
\]
Then
\[
\hat{\zeta}^{\,*}_{z,p_r}\xrightarrow{P}\zeta_p,
\qquad
\hat{\zeta}^{\,*}_{z,p_r}-\zeta_p=O_p(n^{-1/2}).
\]
Moreover, the pooled empirical quantile $Q_n(p)$ satisfies
\[
\sqrt{n}\,\{Q_n(p)-\zeta_p\}
\overset{\mathcal{L}}{\longrightarrow}
N\!\left(0,\frac{\sigma_{p,[\cdot]}^2}{f(\zeta_p)^2}\right),
\qquad
\sigma_{p,[\cdot]}^2
=
\frac1k\sum_{s=1}^k q_s(1-q_s)
=
\sigma_p^2+O(\bar\varepsilon).
\]
In particular, if $\bar\varepsilon\to 0$, the pooled empirical quantile attains the
same first-order asymptotic variance as under perfect ranking.
\end{proposition}

\begin{proof}
Under the mixture identity,
$
E[\hat{F}_n(y)]
=
\frac1k\sum_{s=1}^k F_{[s]}(y)
=
F(y)
$,
$y\in\mathbb R$. Hence the pooled empirical quantile $Q_n(p)=\hat{F}_n^{-1}(p)$ continues to target
$
\zeta_p=F^{-1}(p).
$
The same representation as in Theorem 1 gives
\[
\hat{\zeta}^{\,*}_{z,p_r}
=
\int_0^1 Q_n(u)\,dH^{(z)}_{n,r}(u)+R^{(z)}_{n,r},
\]
with $R^{(HD)}_{n,r}=0$ and $R^{(LF)}_{n,r}=O_p(n^{-1})$. Since the weighting
measures remain deterministic and concentrate around $u=p$, the same argument as
in Theorem 1 yields
\[
\hat{\zeta}^{\,*}_{z,p_r}\xrightarrow{P}\zeta_p,
\qquad
\hat{\zeta}^{\,*}_{z,p_r}-\zeta_p=O_p(n^{-1/2}).
\]

At the fixed point $y=\zeta_p$, the pooled EDF satisfies
\[
\hat{F}_n(\zeta_p)
=
\frac1n
\sum_{s=1}^k\sum_{j=1}^m
1\{Y_{[s]j}\le \zeta_p\},
\]
where the summands in judgment stratum $s$ are Bernoulli with mean
$
q_s=F_{[s]}(\zeta_p).
$
Therefore
\[
E[\hat{F}_n(\zeta_p)] = p,
\qquad
n\,\mathrm{Var}\{\hat{F}_n(\zeta_p)\}
=
\frac1k\sum_{s=1}^k q_s(1-q_s)
=
\sigma_{p,[\cdot]}^2.
\]
By the Lindeberg--Feller theorem,
\[
\sqrt{n}\,\{\hat{F}_n(\zeta_p)-p\}
\overset{\mathcal{L}}{\longrightarrow}
N(0,\sigma_{p,[\cdot]}^2).
\]
Since $f(\zeta_p)>0$, the quantile delta method gives
\[
\sqrt{n}\,\{Q_n(p)-\zeta_p\}
\overset{\mathcal{L}}{\longrightarrow}
N\!\left(0,\frac{\sigma_{p,[\cdot]}^2}{f(\zeta_p)^2}\right).
\]
Finally, because
$
|q_s-p_s|
=
|F_{[s]}(\zeta_p)-F_{(s)}(\zeta_p)|
\le \bar\varepsilon,
$
we have
\[
\sigma_{p,[\cdot]}^2-\sigma_p^2
=
\frac1k\sum_{s=1}^k
\Bigl(q_s(1-q_s)-p_s(1-p_s)\Bigr)
=
O(\bar\varepsilon),
\]
which proves the stated variance comparison.
\end{proof}
\section{Monte Carlo Comparisons}
\label{sec5}

In this section we perform  Monte Carlo simulations to  study the performance of  eight estimators: the empirical quantile under SRS, denoted
SRS(EMP); the empirical quantile applied to the pooled RSS sample, denoted
RSS(EMP) \citep{chen2000ranked}; the Stigler LF and Harrell--Davis estimators
under SRS, denoted SRS(LF) and SRS(HD); the pooled transformed-scale RSS component estimators of
Section~\ref{sec4}, denoted RSS(LF) and RSS(HD); and the two pooled ORSS
L-estimators, denoted ORSS(LF) and ORSS(HD). Three RSS designs are considered:
$(m,k)\in\{(5,3),(5,5),(10,5)\}$, yielding total sample sizes $n\in\{15,25,50\}$.
The target quantile levels are $p\in\{0.1,0.2,\ldots,0.9\}$, and all results
are based on $B=100{,}000$ Monte Carlo replicates.

Three parent distributions are used, chosen to span the main shapes encountered
in environmental and biomedical applications: the standard normal $N(0,1)$
(symmetric, light tails), the standard exponential $\operatorname{Exp}(1)$
(right-skewed, moderate tails), and the Weibull distribution with shape $2$ and
scale $1$, denoted $G(2)$ (right-skewed, lighter upper tail than exponential).

Imperfect ranking is modelled through a concomitant variable
\[
X = \rho\,\frac{Y-\mu_Y}{\sigma_Y}+\sqrt{1-\rho^2}\,Z,
\qquad Z\sim N(0,1)\text{ independent of }Y,
\]
where $(\mu_Y,\sigma_Y)$ denote the population mean and standard deviation of $Y$
and $\rho\in\{1.00,\,0.75,\,0.50\}$ controls ranking quality. Within each cycle
of $k$ units, the $k$ units are ranked by $X$ and the unit assigned judgment
rank $r$ is selected for measurement, simulating the practical scenario in which
a surrogate variable is observed cheaply while the variable of interest is costly
to measure. Perfect ranking ($\rho=1$) recovers the theoretical ideal; moderate
ranking ($\rho=0.75$) corresponds to a Kendall rank correlation between $X$ and
$Y$ of approximately $0.54$; and weak ranking ($\rho=0.50$) corresponds to a
Kendall rank correlation of approximately $0.33$. The relative efficiency of an
estimator $\widehat{\zeta}$ at level $p$ is defined as
\[
\operatorname{RE}(p)
=
\frac{\operatorname{MSE}\!\left(\widehat{\zeta}_{EM,SRS}\right)}
     {\operatorname{MSE}\!\left(\widehat{\zeta}\right)},
\]
where $\operatorname{MSE}$ is computed over the $B=100{,}000$ replicates and the
numerator uses SRS(EMP) as the common reference. Values $\operatorname{RE}(p)>1$
indicate improvement over the SRS baseline.

As noted in Section~\ref{sec4}, the theoretical ORSS weights $\psi_{r_p}(i/n)$ and
$\widetilde{W}^*_{n,i,p}$ are functions of $(k,m,p)$ only  and are  pre-tabulable
for any fixed design. In the simulation, however, a natural implementation
evaluates the polynomial algorithm  at the observed order-statistic values of each
realized RSS sample as discussed in  Section~\ref{subsec:poly}, so that $G_i$ and $g_i$ are assessed on a data-dependent
grid that varies across replicates. Under this implementation strategy, the
per-replicate cost of $O(n^2)$ polynomial operations must be repeated for each of
the $B=100{,}000$ replicates. For $(m,k)=(10,5)$, this cost remains computationally
demanding, and we therefore restrict the ORSS estimators to
$(m,k)\in\{(5,3),(5,5)\}$ in the simulation study. We emphasize that with a fixed pre-tabulated
probability grid one could avoid per-replicate recomputation entirely; the
restriction here reflects a choice of simulation implementation, not a fundamental
limitation of the estimators.

\begin{figure}
\begin{center}
\includegraphics[scale=1.1]{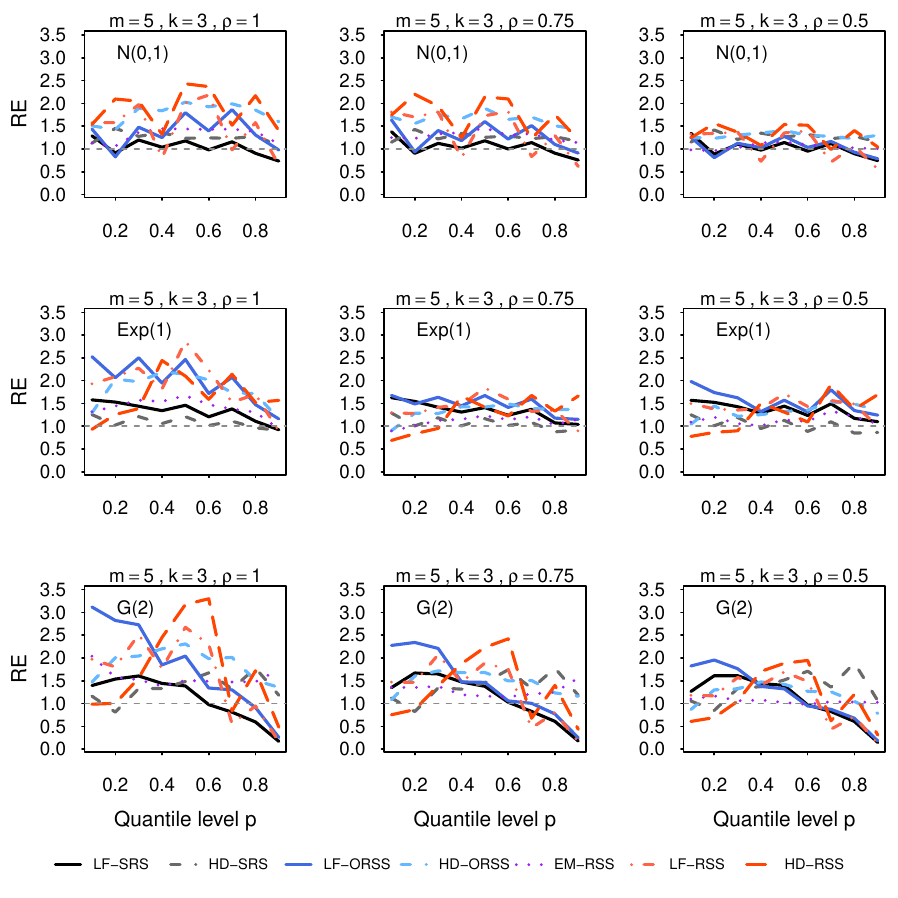}
\caption{\label{pic1}Relative efficiency (RE) of eight quantile estimators
under simple random sampling (SRS) and ranked set sampling (RSS), evaluated
against the SRS empirical quantile estimator SRS(EMP), for $(m,k)=(5,3)$
(total sample size $n=15$). Rows correspond to the parent distributions
$N(0,1)$, $\operatorname{Exp}(1)$, and $G(2)$ (Weibull with shape~2 and
scale~1). Columns correspond to ranking quality
$\rho\in\{1.00,\,0.75,\,0.50\}$ (perfect, moderate, weak). The eight
estimators are SRS(EMP) (reference, shown as a dashed grey line at
$\operatorname{RE}=1$), SRS(LF), SRS(HD), RSS(EMP), RSS(LF), RSS(HD),
ORSS(LF), and ORSS(HD). Values $\operatorname{RE}>1$ indicate improvement
over the SRS baseline. All results are based on $B=100{,}000$ Monte Carlo
replicates.}
\end{center}
\end{figure}

Figure~\ref{pic1} presents results for $(m,k)=(5,3)$ (total sample size $n=15$).
Under perfect ranking ($\rho=1$) and the normal distribution, the proposed RSS
L-estimators dominate throughout: RSS(HD) achieves the highest relative efficiency
across the full range of quantile levels, reaching approximately $2.5$ at $p=0.5$,
with ORSS(HD) performing comparably. Both substantially exceed the efficiency of
RSS(EMP), SRS(LF), and SRS(HD). Under Exp(1) and $G(2)$, the picture is more
nuanced: RSS(LF) leads at lower quantile levels ($p\leq 0.3$), reflecting the
ability of the locally concentrated LF weights to track the mass in the lower tail
of right-skewed distributions, while RSS(HD) takes over for $p\geq 0.5$ and
achieves high  performance at upper quantiles. The ORSS estimators generally
fall between the RSS L-estimators and the empirical RSS quantile.  All RSS-based
L-estimators remain above the SRS baseline ($\operatorname{RE}>1$) for essentially
all quantile levels across all three distributions under perfect ranking. As ranking
quality deteriorates to $\rho=0.75$, the efficiency curves shift downward, but
RSS(HD) and RSS(LF) retain clear advantages over SRS(EMP) across most of the
quantile range. Under weak ranking ($\rho=0.50$), all RSS estimators remain at or
above the SRS baseline, with RSS(HD) still achieving relative efficiencies of
approximately $1.5$--$1.6$ at central quantiles under the normal distribution.
This robustness to imperfect ranking is consistent with the qualitative MSE
continuity established in Proposition~\ref{prop:imperfect-rate}.

\begin{figure}
\begin{center}
\includegraphics[scale=1.1]{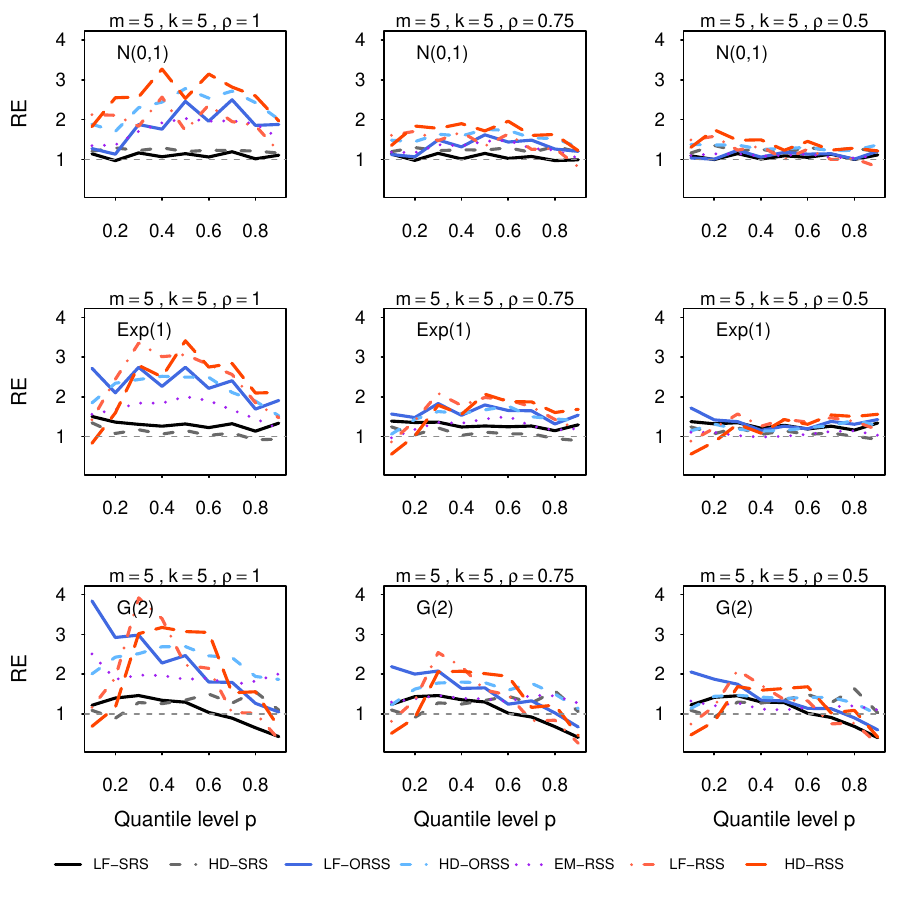}
\caption{\label{pic2}Relative efficiency (RE) of eight quantile estimators
under SRS and RSS for $(m,k)=(5,5)$ (total sample size $n=25$). Rows,
columns, and estimator labels are as in Figure~\ref{pic1}. ORSS estimators
are included because $n=25$ remains within the range where the
polynomial-multiplication algorithm of Remark~\ref{rem:poly} makes their
computation tractable.}
\end{center}
\end{figure}

Figure~\ref{pic2} shows results for $(m,k)=(5,5)$ (total sample size $n=25$).
Increasing the set size from $k=3$ to $k=5$ produces a substantial improvement for
all RSS estimators, reflecting the greater between-strata coverage of the population
distribution when more rank strata are used. Under perfect ranking and the normal
distribution, RSS(HD) achieves relative efficiencies approaching $4.0$ at central
quantiles, a roughly $25\%$ improvement over the $(m,k)=(5,3)$ configuration despite
the same number of cycles. ORSS(HD) is competitive, reaching similar peaks. For
Exp(1) and $G(2)$, the qualitative pattern observed at $k=3$ is preserved but
amplified: RSS(LF) dominates at lower quantiles while RSS(HD) is best at central
and upper quantiles. The crossover between RSS(LF) and RSS(HD) occurs near $p=0.4$
for Exp(1) and near $p=0.3$ for $G(2)$, where the Weibull mass is more
concentrated. Under moderate ranking ($\rho=0.75$), RSS(HD) still achieves relative
efficiencies exceeding $2.0$ at central quantiles under the normal distribution, and
both RSS L-estimators comfortably outperform RSS(EMP) and the SRS-based estimators.
Under weak ranking ($\rho=0.50$), RSS(HD) and RSS(LF) maintain relative efficiencies
above $1.5$ at central quantiles, confirming the practical robustness of the method
to the level of ranking quality typically encountered in applications.


\begin{figure}[h!]
\begin{center}
\includegraphics[scale=1.1]{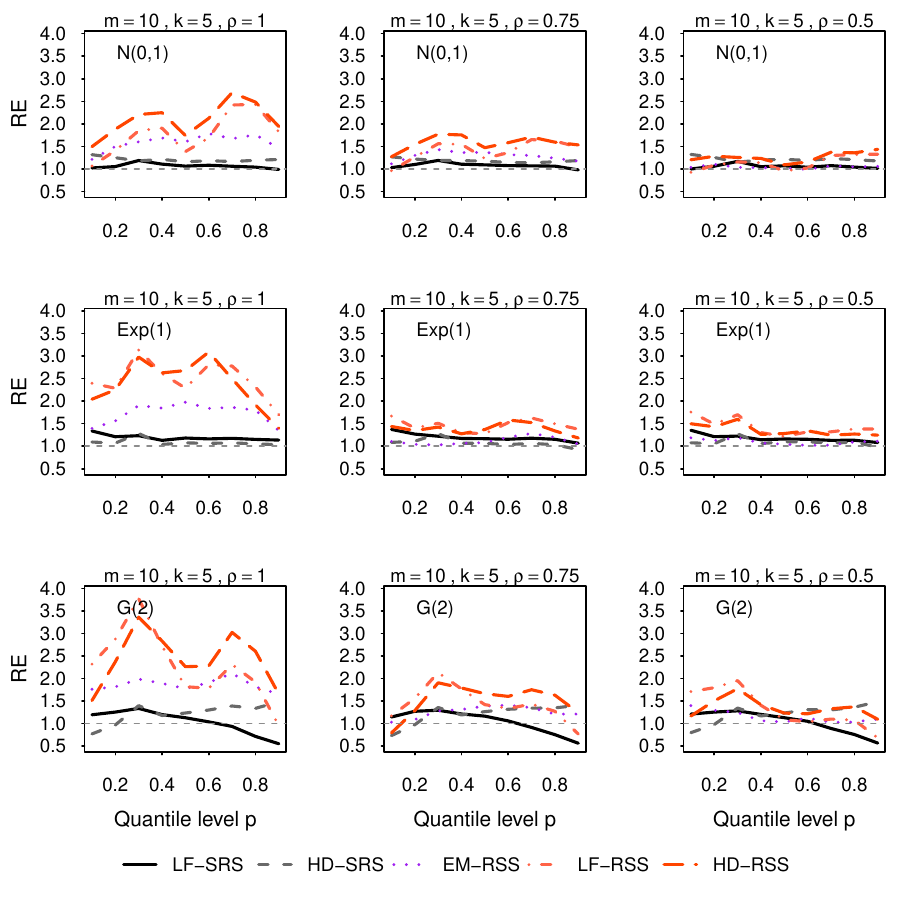}
\caption{\label{pic3}Relative efficiency (RE) of six quantile estimators
under SRS and RSS for $(m,k)=(10,5)$ (total sample size $n=50$). Rows and
columns are as in Figure~\ref{pic1}. ORSS estimators are excluded from this
configuration; see the discussion in Section~\ref{sec5}.}
\end{center}
\end{figure}

Figure~\ref{pic3} presents results for $(m,k)=(10,5)$ (total sample size $n=50$),
where the ORSS estimators are excluded for the reasons discussed above. The larger
sample size provides additional stabilization, and the RSS L-estimators demonstrate
their strongest performance across the study.  Even under weak ranking ($\rho=0.50$), the RSS L-estimators
outperform SRS(EMP) for all quantile levels and all distributions considered,
confirming that the efficiency gains of the proposed methodology are robust to
practical imperfections in the ranking mechanism.

Taken together, the simulation results establish three main conclusions. First,
RSS(HD) is the best overall estimator for moderate and upper quantiles under all
three distributions: it delivers the largest efficiency gains under symmetric
distributions and remains competitive under skewed ones. Second, RSS(LF) is
preferable at lower quantile levels of right-skewed distributions, where its more
locally concentrated weighting scheme is better adapted to the tail behavior of
Exp(1) and $G(2)$. Third, the efficiency advantages of the proposed estimators
scale with the set size $k$ and are remarkably robust to imperfect ranking: even
at $\rho=0.50$, a level of ranking quality that is modest by practical standards,
both RSS(LF) and RSS(HD) maintain clear improvements over the SRS baseline across
all configurations studied.


\section{NHANES Liver Transient Elastography}
\label{sec6}

A practically appealing biomedical application of the proposed methodology arises from liver transient elastography data in the National Health and Nutrition Examination Survey (NHANES). The August 2021--August 2023 NHANES release provides public documentation for the overall cycle, the examination component, the liver ultrasound transient elastography file ({LUX\_L}), the body-measures file ({BMX\_L}), and the demographics file ({DEMO\_L}) \citep{nhanes2021overview,nhanes2021exam,nhanes2021lux,nhanes2021bmx,nhanes2021demo}. In this cycle, liver transient elastography was performed in the mobile examination center, and the public file records two summary outcomes representing  the median liver stiffness measurement, denoted {LUXSMED} and measured in kilopascals (kPa), and the median controlled attenuation parameter, denoted {LUXCAPM} and measured in dB/m \citep{nhanes2021lux}. In this setting, measuring the response is comparatively costly and specialized, whereas a concomitant ranking variable can be obtained quickly and inexpensively from routine anthropometric assessment. From a clinical perspective, liver stiffness is widely used as a noninvasive surrogate of fibrosis severity, while CAP is used to assess hepatic steatosis \citep{karlas2017cap,eddowes2019accuracy}. Accordingly, quantiles in the middle and upper parts of these distributions are of direct interest for screening, risk stratification, and the identification of clinically elevated subpopulations.

In our illustration, the observed NHANES adult sample is treated as a finite population from which repeated RSS samples are drawn. We restrict attention to adults because clinically meaningful interpretation of liver stiffness and steatosis quantiles is more naturally framed in the context of adult metabolic dysfunction-associated liver disease, obesity, and fibrosis risk, and because anthropometric rankers such as waist circumference and body mass index have a more stable clinical interpretation in adults than in children or adolescents. After restricting the August 2021--August 2023 NHANES file to adults aged 20 years or older with observed elastography outcomes and available ranking variables, the resulting analytic finite population consisted of \(N= 5,614\) participants. Within each set of size $k$, individuals are ranked by a cheap concomitant variable and only the selected unit is ``measured'' for the costly response, thereby mimicking the practical situation in which anthropometric variables are easy to obtain but FibroScan measurements are not. We emphasize that NHANES complex survey weights are not used in this illustration, because the purpose is methodological comparison under a controlled RSS design rather than design-based population inference for the U.S. population. We considered two natural candidate rankers, waist circumference ({BMXWAIST}) and body mass index ({BMXBMI}), and evaluated the estimators on the finer grid $p\in\{0.20, 0.25, \ldots, 0.80\}$. The number of RSS replicates was set to $B=100{,}000$.

The exploratory analysis already supports this setup. For {LUXSMED}, the empirical distribution is strongly right-skewed, with median $5.0$, interquartile range $[4.1,6.3]$, and maximum $75.0$, whereas {LUXCAPM} has median $250$, interquartile range $[210,301]$, and range $100$ to $400$. The scatterplots also show a clear positive association between the elastography outcomes and anthropometric variables. For {LUXSMED}, the rank-correlation screen slightly favors BMI over waist circumference (Spearman correlations $0.330$ versus $0.351$, respectively, with BMI selected as the preferred ranking variable in the final illustration); for {LUXCAPM}, waist circumference is the preferred ranker.

\begin{table}[htbp]
\centering
\caption{Compact descriptive summary for the two NHANES elastography outcomes used in the real-data illustration.}
\label{tab:nhanes-desc}
\begin{tabular}{lrrrrrrrr}
\hline
Outcome & Mean & SD & Min & Q1 & Median & Q3 & Max & Preferred ranker \\
\hline
{LUXSMED} & 6.19 & 6.13 & 1.5 & 4.1 & 5.0 & 6.3 & 75.0 & {BMXBMI} \\
{LUXCAPM} & 251.93 & 64.09 & 100.0 & 210.0 & 250.0 & 301.0 & 400.0 & {BMXWAIST} \\
\hline
\end{tabular}
\end{table}

\begin{figure}[htbp]
\centering
\begin{minipage}{0.48\textwidth}
\centering
\includegraphics[page=1,width=\textwidth]{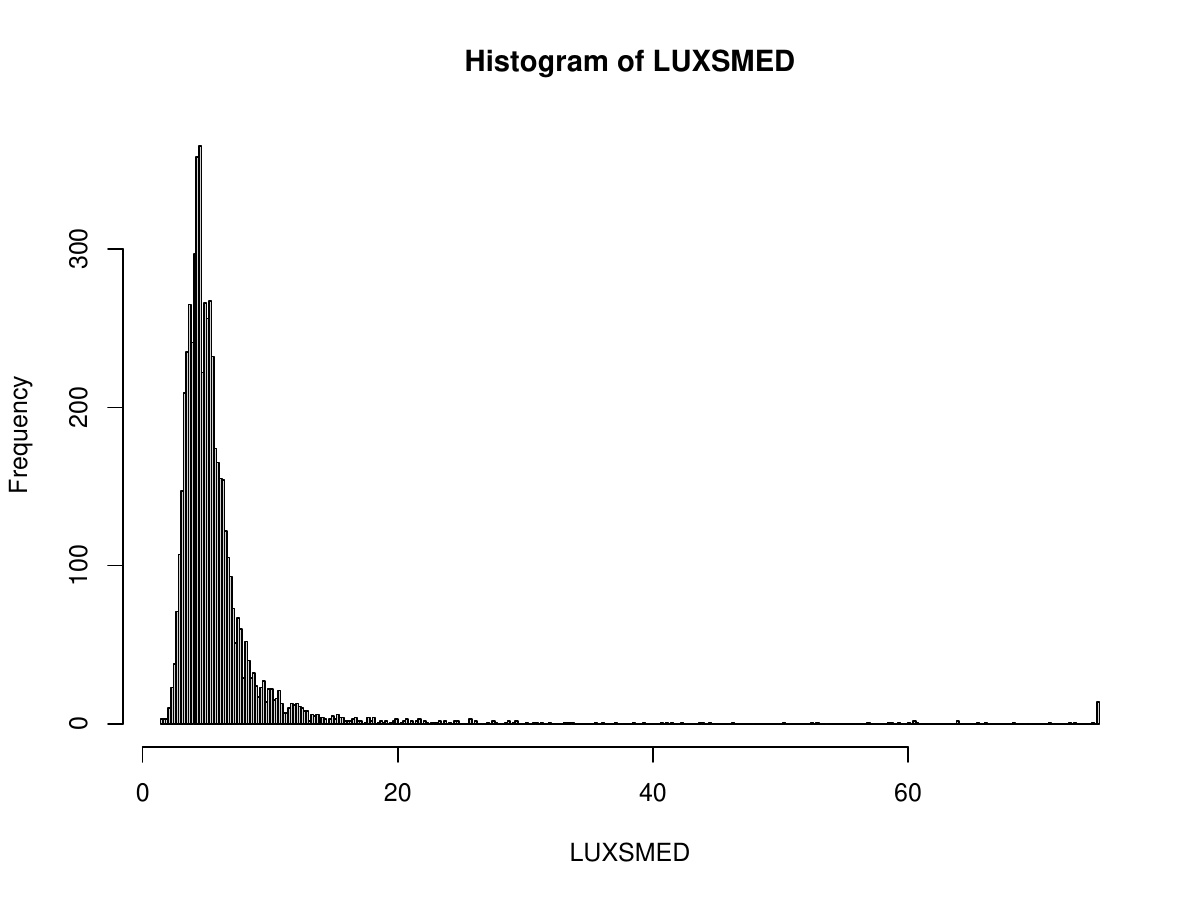}
\vspace{1mm}
\small (a) Histogram of {LUXSMED}
\end{minipage}\hfill
\begin{minipage}{0.48\textwidth}
\centering
\includegraphics[page=3,width=\textwidth]{nhanes_LUXSMED_m5k3_B100000_EDA_plots.pdf}
\vspace{1mm}
\small (b) BMI versus {LUXSMED}
\end{minipage}

\vspace{3mm}
\begin{minipage}{0.48\textwidth}
\centering
\includegraphics[page=1,width=\textwidth]{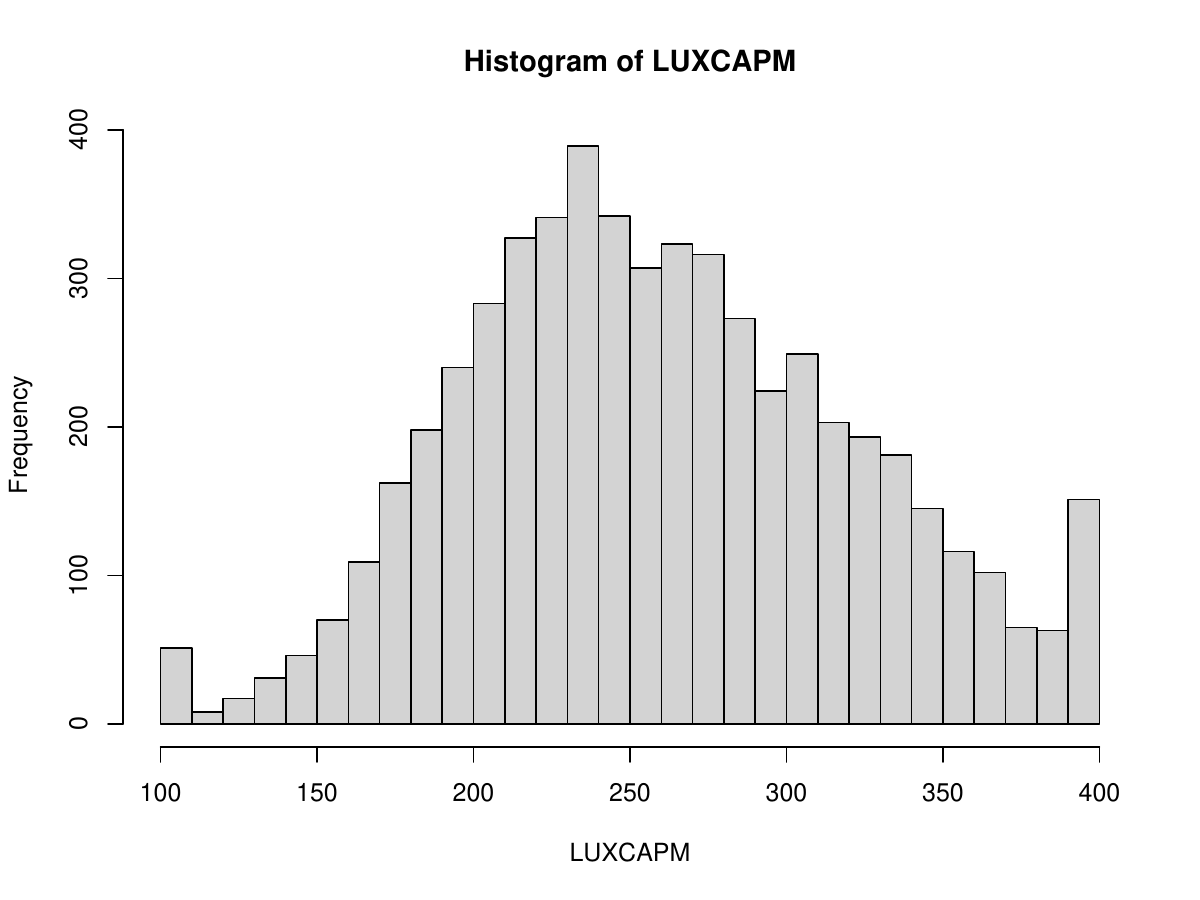}
\vspace{1mm}
\small (c) Histogram of {LUXCAPM}
\end{minipage}\hfill
\begin{minipage}{0.48\textwidth}
\centering
\includegraphics[page=3,width=\textwidth]{nhanes_LUXCAPM_m5k3_B100000_EDA_plots.pdf}
\vspace{1mm}
\small (d) Waist circumference versus {LUXCAPM}
\end{minipage}
\caption{Exploratory plots for the NHANES liver transient elastography illustration. The upper panels show the heavy right tail of {LUXSMED} and its positive association with BMI. The lower panels show the broader distribution of {LUXCAPM} and the corresponding positive association with waist circumference.}
\label{fig:nhanes-eda}
\end{figure}

Figure~\ref{fig:nhanes-re} summarizes the relative-efficiency comparisons for the real-data RSS study with $(m,k)=(5,3)$. The left panel reports the {LUXSMED} analysis using BMI as the ranking variable, while the right panel reports the {LUXCAPM} analysis using waist circumference. The general message is consistent with the simulation study: the RSS-based L-estimators frequently dominate the empirical SRS benchmark and often improve on the SRS L-estimators as well, though the two outcomes display markedly different patterns.

For {LUXSMED}, the picture is driven primarily by the heavy right skew of the distribution and the moderate correlation of BMI with liver stiffness. The LF-type estimators---{RSS(LF)} and {ORSS(LF)}---are the strongest performers across the middle of the quantile range, with {RSS(LF)} peaking at approximately $1.6$ near $p=0.55$ and {ORSS(LF)} reaching similar levels near $p=0.5$--$0.55$. {SRS(LF)} and {ORSS.HD} maintain moderate, broadly stable efficiency gains of roughly $1.2$--$1.35$ over most of the range. By contrast, {RSS(HD)} and {SRS(HD)} exhibit an extreme collapse in the upper tail: both methods deteriorate sharply above $p=0.65$ and fall well below one by $p=0.75$--$0.80$, with {RSS(HD)} reaching values near $0.2$--$0.25$ at these upper quantile levels. This pattern reflects the sensitivity of the HD weighting scheme to the long upper tail of the liver stiffness distribution when combined with a moderate-quality ranker, and reinforces the practical importance of ranker selection for heavily skewed outcomes. {RSS(EMP)} remains near or just above one throughout, providing a stable but unimpressive baseline. Thus, for liver stiffness quantile estimation with a moderate ranker, the LF-type estimators offer the most reliable efficiency gains, while HD-type methods should be used with caution in the upper tail.

\begin{figure}[h!]
\centering
\begin{minipage}{0.48\textwidth}
\centering
\includegraphics[width=\textwidth]{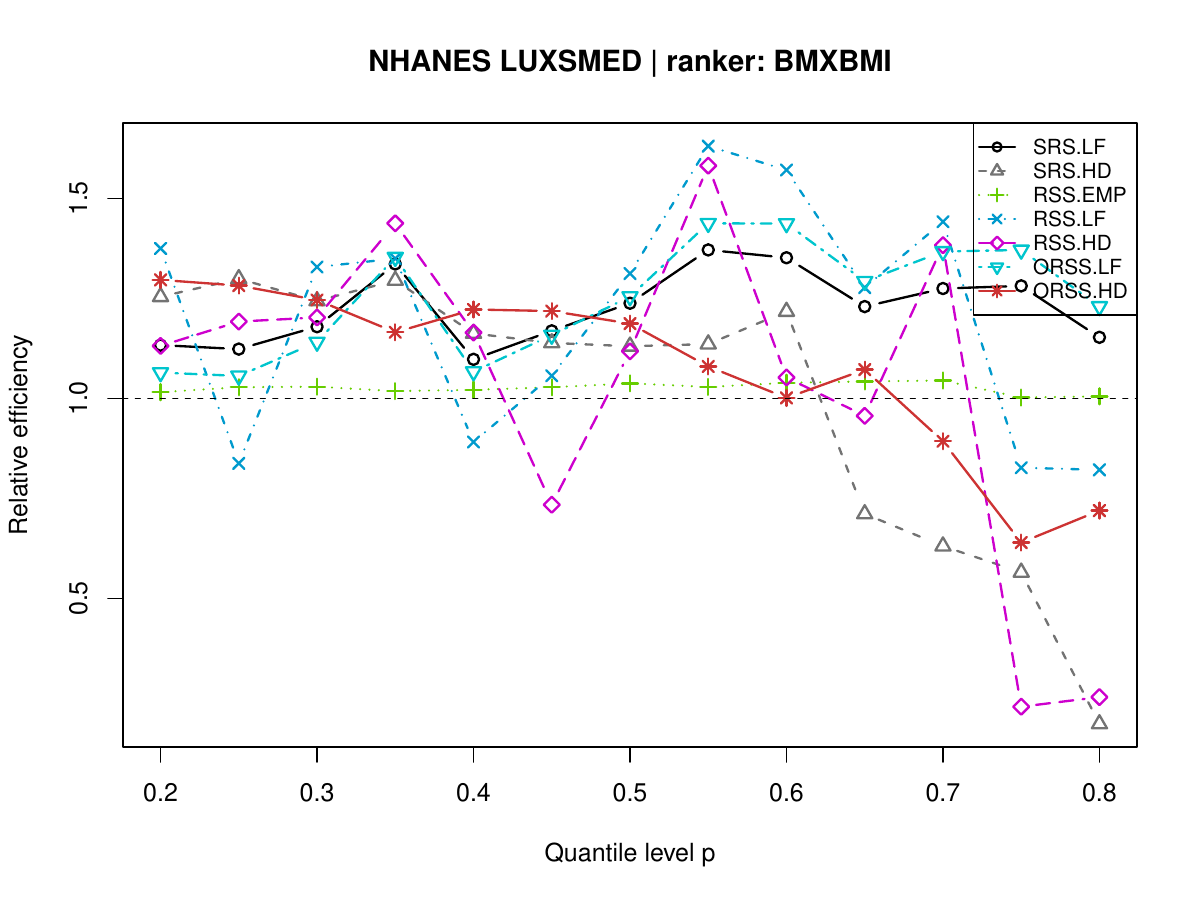}
\vspace{1mm}
\small (a) Relative efficiency for {LUXSMED} (ranker: {BMXBMI})
\end{minipage}\hfill
\begin{minipage}{0.48\textwidth}
\centering
\includegraphics[width=\textwidth]{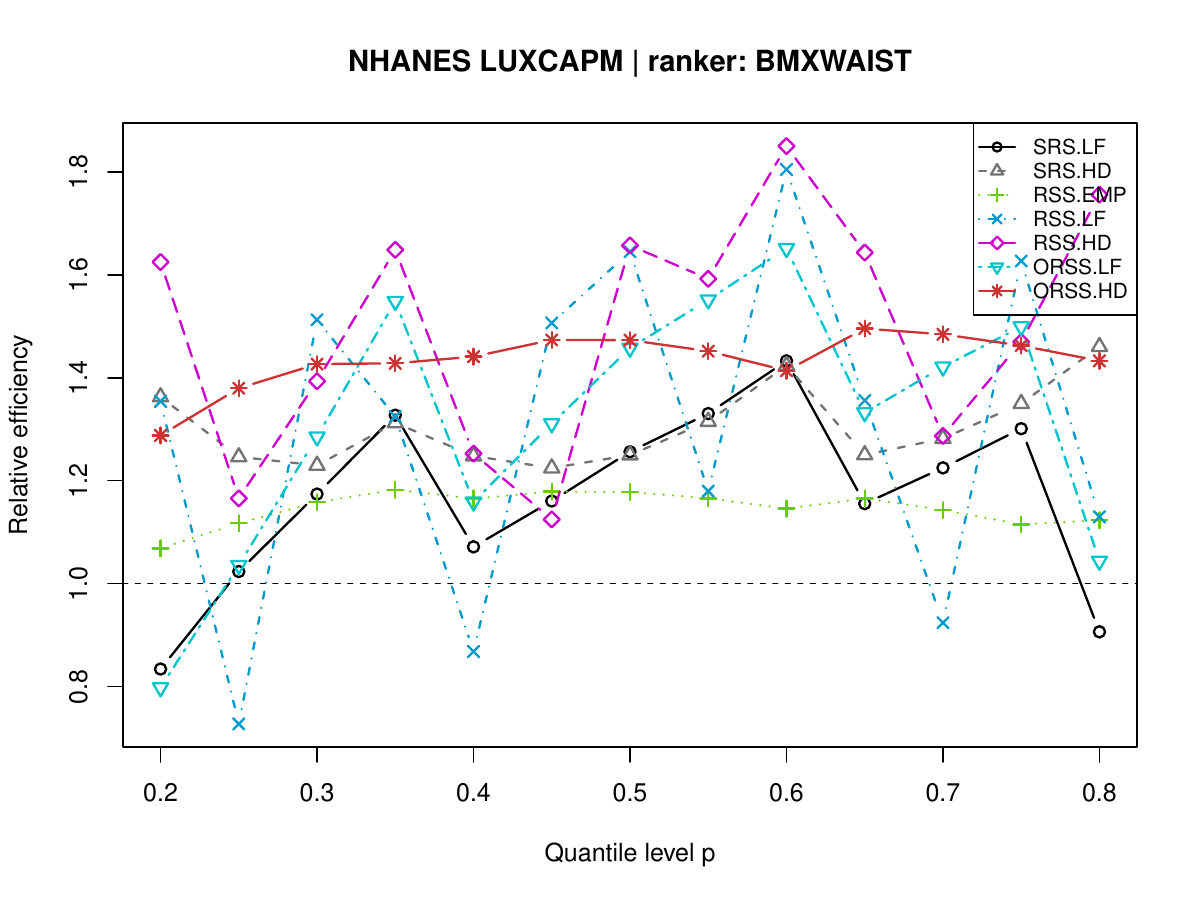}
\vspace{1mm}
\small (b) Relative efficiency for {LUXCAPM} (ranker: {BMXWAIST})
\end{minipage}
\caption{Relative-efficiency curves for the NHANES RSS illustration with $(m,k)=(5,3)$ and $B=100{,}000$ replicates, evaluated at $p\in\{0.20,0.25,\ldots,0.80\}$. Values greater than one indicate improvement over the empirical SRS quantile estimator. The left panel uses BMI ({BMXBMI}) as the ranking variable for {LUXSMED}; the right panel uses waist circumference ({BMXWAIST}) for {LUXCAPM}. For {LUXSMED}, {RSS(LF)} and {ORSS(LF)} are the dominant methods in the central quantile range, while {RSS(HD)} and {SRS(HD)} collapse in the upper tail. For {LUXCAPM}, {RSS(HD)} achieves the largest gains and {ORSS(HD)} provides the most stable efficiency profile across all reported quantile levels.}
\label{fig:nhanes-re}
\end{figure}

For {LUXCAPM}, the evidence in favor of the RSS-based methods is strong and covers a wider range of quantile levels. {RSS(HD)} is the most prominent performer, displaying large peaks at several quantile levels and reaching approximately $1.85$ near $p=0.6$ and $1.8$ near $p=0.65$; it remains above one across virtually the entire quantile range despite some oscillation at lower levels. {ORSS(HD)} is the most stable estimator in this analysis, maintaining relative efficiencies consistently in the range $1.3$--$1.5$ across all reported quantile levels, confirming the reliability of the ORSS construction when the polynomial-algorithm weights are well-calibrated to the design. {RSS(LF)} also achieves substantial gains at many quantile levels but is more variable, with dips near $p=0.25$ and $p=0.4$ where it falls briefly toward or below one. The SRS-based methods, {SRS(LF)} and {SRS(HD)}, show moderate gains through most of the range but deteriorate at the upper tail, with {SRS(LF)} dropping below one near $p=0.8$. {RSS(EMP)} provides modest but consistent gains above one throughout. Overall, for the controlled attenuation parameter both the HD and LF variants of the proposed RSS estimators deliver clear practical advantages, with {ORSS(HD)} offering the most stable efficiency profile across the full quantile grid.

\section{Discussion}
\label{sec7}

We have developed a framework for L-estimation of population quantiles under ranked set sampling. The core methodological contribution is the stratum-based family of RSS L-estimators (RSS(LF) and RSS(HD)), derived from a quantile decomposition identity that reduces the RSS quantile problem to $k$ independent SRS-type subproblems, one per rank stratum. The formal theory in the paper is strongest at the stratum level, where consistency and asymptotic normality can be established under standard regularity conditions. For the final combined estimator, these results provide a principled first-order motivation, while a complete asymptotic treatment of the ordering step is left for future work. Empirically, the Harrell--Davis variant tends to perform best for moderate and upper quantiles, whereas RSS(LF) can be competitive for lower quantiles under right-skewed distributions. Both estimators are computationally scalable and easy to implement.

Several directions for future research arise from this work. The stratum-level asymptotic variance expressions derived in Section~\ref{sec4} suggest natural criteria for choosing $(m,k)$ or for designing unbalanced RSS schemes targeted at a quantile level $p$, extending the optimal design perspective of \cite{chen2001optimal} to the L-estimation framework. Developing consistent plug-in or bootstrap variance estimators for the final combined estimator would enable confidence interval construction and would complement the results of \cite{nourmohammadi2014confidence}. The methodology can be extended to other RSS-type designs including median-ranked set sampling, and partially rank-ordered set sampling \citep{hatefi2017improved}, all of which produce independent non-identically distributed observations for which the stratum decomposition idea may apply. Finally, since L-estimators of quantiles are closely related to distortion risk measures and inequality indices such as the Gini coefficient, the RSS-based ideas developed here may also improve the precision of these quantities in surveys that employ RSS-type stratification.

\section*{Acknowledgment}
Mohammad Jafari Jozani gratefully acknowledges partial support from NSERC Canada. The authors are also grateful to Dr.\ Saeid Amiri for his valuable comments, helpful suggestions, and insightful discussions, which greatly improved this work.

\bibliographystyle{apalike}
\bibliography{ref_paper}

\end{document}